\documentclass[twocolumn,twosides]{IEEEtran} 

\makeatletter
\let\NAT@parse\undefined
\makeatother

\usepackage{pstricks}
\usepackage{lipsum}
\usepackage[numbers,sort&compress]{natbib}
   
\usepackage{amssymb}
\usepackage{pgfplots}
\usepackage{graphicx}
\pgfplotsset{compat=1.3}
\usepackage[cmex10]{amsmath}
\usepackage{color}
\usepackage{flushend}
\usepackage{multirow}

\newtheorem{lemma}{Lemma}
\newtheorem{definition}{Definition}

\makeatletter
\def\blfootnote{\xdef\@thefnmark{}\@footnotetext}
\makeatother

\begin{document}

\title{\huge{An Extension of the $\kappa$-$\mu$ Shadowed Fading Model: Statistical Characterization and Applications}}

\author{Pablo Ramirez-Espinosa, F. Javier Lopez-Martinez, Jose F. Paris, Michel D. Yacoub, Eduardo Martos-Naya}


\maketitle
\blfootnote{\noindent P. Ramirez-Espinosa, F. J. Lopez-Martinez, J. F. Paris and E. Martos-Naya are with Departmento de Ingenier\'ia de Comunicaciones, Universidad de Malaga - Campus de Excelencia Internacional Andaluc\'ia Tech., Malaga 29071, Spain. E-mail: \{pre,fjlopezm,paris,eduardo\}@ic.uma.es. This work has been funded by the Consejer\'ia de Econom\'ia, Innovaci\'on, Ciencia y Empleo of the Junta de Andaluc\'ia, the Spanish Government
and the European Fund for Regional Development FEDER (projects P2011-TIC-7109, P2011-TIC-8238, and TEC2014-57901-R).
 \\ \indent M. D. Yacoub is with the Wireless Technology Laboratory, School of Electrical and Computer Engineering, University of Campinas, Campinas 13083-970, Brazil. E-mail michel@decom.fee.unicamp.br.\\
\indent This work has been submitted to the IEEE for publication. Copyright may be transferred without notice, after which this version may no longer be accesible.
}

\begin{abstract}
We here introduce an extension and natural generalization of both the $\kappa$-$\mu$ shadowed and the classical Beckmann fading models: the Fluctuating Beckmann (FB) fading model. This new model considers the clustering of multipath waves on which the line-of-sight ({{LoS}}) components randomly fluctuate, together with the effect of in-phase/quadrature power imbalance in the {{LoS}} and non-{{LoS}} components. Thus, it unifies a variety of important fading distributions as the one-sided Gaussian, Rayleigh, Nakagami-$m$, Rician, $\kappa$-$\mu$, $\eta$-$\mu$, $\eta$-$\kappa$, Beckmann, Rician shadowed and the $\kappa$-$\mu$ shadowed distribution. The chief probability functions of the FB fading model, namely probability density function, cumulative distribution function and moment generating function are {derived}. The second-order statistics such as the level crossing rate and the average fade duration are also analyzed. {{These results can be used to}} derive some performance metrics of interest of wireless communication systems operating over FB fading channels. 
\end{abstract}
\begin{keywords}
Fading channels, Beckmann, Rayleigh, Nakagami-$m$, Rician, $\kappa$-$\mu$, 
Rician shadowed, $\kappa$-$\mu$ shadowed.
\end{keywords}

\section{Introduction}

{{}In wireless environments, the radio signal is affected by a number of random phenomena including reflection (both specular and diffuse), diffraction, and scattering as they travel from transmitter to receiver, giving rise to the so-called multipath propagation. At the receiver, the resulting signal appears as a linear combination of the multipath waves, each of which with their own amplitudes and phases.} When the number of paths is sufficiently large, the complex baseband signal can be regarded as Gaussian because of the Central Limit Theorem (CLT). Depending on the choice of {{the}} parameters characterizing this complex Gaussian random variable, namely the mean and variance of the in-phase and quadrature components, different fading models emerge: Rayleigh (zero-mean and equal variances), Hoyt (zero-mean and unequal variances) and Rice (non-zero mean, equal variances), which are perhaps the most popular fading models arising from the CLT assumption \cite{Rice1944,Hoyt1947}. 

The most general case (i.e. unequal means and variances for the in-phase and quadrature components) was considered by Beckmann \cite{Beckmann1962,Beckmann1964} when characterizing the scattering from rough surfaces. However, its {{greater flexibility}} comes at the price of an increased mathematical complexity; {in fact, its chief probability functions (PDF and CDF) are known to be given in infinite-series form 
expression \cite{AlouiniBook}}, as opposed to Rayleigh, Hoyt and Rician models. {{}Other models characterizing the joint effects of imbalances in the mean and variance between in-phase and quadrature components whose PDF and CDF are given in infinite-series form are the so called $\eta$-$\kappa$ \cite{Yacoub05,Yacoub05a} and the very recently proposed $\alpha$-$\eta$-$\kappa$-$\mu$ \cite{Yacoub2016}}

In order to provide a better statistical characterization of the received radio signal in multipath environments, some alternative models have been proposed as generalizations of classical Rayleigh, Hoyt and Rician. By means of considering the effect of clustering of multipath waves, two new fading models arise \cite{Yacoub07}: the $\eta$-$\mu$ fading model as a generalization of Hoyt model, well-suited for non line-of-sight ({{NLoS}}) propagation environments, and the $\kappa$-$\mu$ fading model as a generalization of Rice model in line-of-sight\footnote{{{}Line-of-sight is used here to mean the more precise phenomenon concerning the presence of dominant components.}} ({{LoS}}) scenarios. These models have become of widespread use in the recent years because of their flexibility and relatively simple mathematical tractability, as their chief probability functions (PDF, CDF and MGF) are given in closed-form \cite{Yacoub07,Morales2010,Ermolova2008}. Besides, both models also include the versatile and popular Nakagami-$m$ model as particular case \cite{Nakagami1960}.

A further generalization of these models was introduced in \cite{Paris2014} and \cite{Cotton2015} under the name of $\kappa$-$\mu$ shadowed fading distribution. This new distribution provides an additional degree of freedom compared to the $\kappa$-$\mu$ distribution by allowing the {{LoS}} component to randomly fluctuate. Notably, the $\kappa$-$\mu$ shadowed fading model includes both the $\kappa$-$\mu$ and $\eta$-$\mu$ models \cite{Laureano2016} as special cases, as well as the Rician shadowed fading model \cite{Abdi2003}. Thus, most popular fading models in the literature for {{LoS}} and {{NLoS}} conditions are unified under the umbrella of the $\kappa$-$\mu$ shadowed fading channel model. This {{greater flexibility}} does not come at the price of an increased mathematical complexity; in fact, in some cases its PDF and CDF admit a representation in terms of a finite number of powers and exponentials, thus becoming even as tractable as the Nakagami-$m$ distribution \cite{Lopez2017}.

Even though the $\kappa$-$\mu$ shadowed fading model succeeds on capturing different propagation phenomena such as clustering and {{LoS}} fluctuation, it fails when it comes to accounting for the effect of power imbalance in the {{LoS}} and {{NLoS}} components as originally considered by Beckmann \cite{Beckmann1962,Beckmann1964}. Motivated by this issue, in this paper we introduce an extended $\kappa$-$\mu$ shadowed fading model which effectively captures such propagation conditions. This new model can be regarded as a generalization of the original fading model in \cite{Paris2014}, but also as a generalization of Beckmann fading model by also including the effects of clustering and {{LoS}} fluctuation. For this reason, and for the sake of notational brevity, we deem appropriate to name it as the Fluctuating Beckmann (FB) fading model (or equivalently, fading distribution).

The FB model includes as special cases an important set of fading distributions as the one-sided Gaussian, Rayleigh, Nakagami-$m$, Rician, $\kappa$-$\mu$, $\eta$-$\mu$, $\eta$-$\kappa$, Beckmann, Rician shadowed and the $\kappa$-$\mu$ shadowed distributions. 

{{}Interestingly, the CDF and PDF of the FB fading model are given in terms of a well-known function in the context of communication theory, having a functional form similar to the original $\kappa$-$\mu$ shadowed fading model.} The randomization of the {{LoS}} component allows for including an additional degree of freedom when compared to the Beckmann model. We provide a full statistical characterization of the FB fading model in terms of its first-order statistics (PDF, CDF and MGF) and second-order statistics (level crossing rate and average fade duration), and then exemplify its applicability to wireless performance analysis. 

The remainder of this paper is structured as follows: the physical model of the FB fading distribution is described in Section \ref{ModelSection}. In Section \ref{StatSection} the PDF, CDF and MGF of this distribution are derived. Then, in Section \ref{SecondOrder} the level crossing rate (LCR) and average fade duration (AFD) are computed. These statistical results are then used to derive some performance metrics of interest in Section \ref{ApplicationsSection}. Finally, the main conclusions are outlined in Section \ref{ConclusionSection}.

\section{Physical model}
\label{ModelSection}
The physical model of the FB distribution arises as a generalization of the physical model of the $\kappa$-$\mu$ shadowed distribution \cite{Paris2014,Laureano2017}. The received radio {{signal is built out of a superposition}} of radio waves grouped into a number of clusters of waves, and the received signal power $W$ can be expressed in terms of the in-phase and quadrature components of the received signal affected by fading as follows 
\begin{equation}
\label{Model}
		W = \sum_{i=1}^{\mu}\left(X_i + p_i\xi\right)^2 + \left(Y_i + q_i\xi\right)^2
\end{equation}
where $\mu$ is a natural number indicating the number of clusters, $X_i$ and $Y_i$ are mutually independent Gaussian random processes with $E[X_i] = E[Y_i] = 0$, $E[X_i^2] = \sigma_x^2$, $E[Y_i^2] = \sigma_y^2$, $p_i$ and $q_i$ are real numbers and $\xi$ is\footnote{or equivalently, $\xi^2$ is a Gamma random variable with $E[\xi^2] = 1$, shape parameter $m$ and scale parameter $1/m$.} a Nakagami-$m$ distributed random variable with shape parameter $m$ and $E[\xi^2] = 1$ which accounts for the fluctuation of the {{LoS}} component.

As opposed to the $\kappa$-$\mu$ shadowed fading model, we here consider that $X_i$ and $Y_i$ can have different variances. Thus, the effect of power imbalance in the diffuse components associated to non-{{LoS}} propagation is considered. Similarly, we also assume that the power of the {{LoS}} components can be imbalanced, i.e. $p^2\triangleq\sum_{i=1}^{\mu}p^2_i \neq q^2 \triangleq \sum_{i=1}^{\mu} q^2_i$. Hence, the physical model in (\ref{Model}) can be regarded as a generalization of the Beckmann fading model through the consideration of clustering and {{LoS}} fluctuation.
%

%
\section{First Order Statistics}
\label{StatSection}

\begin{figure*}[!t]
\setcounter{equation}{7}
\normalsize
\begin{align}
\label{MGF_fin}
	M_{\gamma}(s)=&\frac{1}{\left(1-\frac{2\eta}{\mu(1+\eta)(1+\kappa)}\bar{\gamma s} \right)^{\mu/2} \left(1-\frac{2}{\mu(1+\eta)(1+\kappa)}\bar{\gamma s} \right)^{\mu/2}}  \notag \\
	&\times\left[1-\frac{1}{m}\left(\frac{\mu\kappa\left( \frac{\varrho^2}{1+\varrho^2}\right)(1+\eta)\bar{\gamma}s}{(1+\eta)(1+\kappa)\mu - 2\eta\bar{\gamma} s} \right.\right.\left.\left.+  \frac{\mu\kappa\left( \frac{1}{1+\varrho^2}\right)(1+\eta)\bar{\gamma}s}{(1+\eta)(1+\kappa)\mu - 2\bar{\gamma} s}\right) \right]^{-m}
\end{align}
\begin{align}\label{eqPDF}
		f_{\gamma}(\gamma) =\frac{\alpha_2^{m-\mu/2} \gamma^{\mu-1}}{\bar{\gamma}^{\mu}\Gamma(\mu)\alpha_1^m}
		\Phi_2^{(6)} \left(\frac{\mu}{2}, \frac{\mu}{2},-m, -m, m, m; \mu;\right.\left.\frac{-\gamma}{\bar{\gamma}\sqrt{\eta\alpha_2}}, 
			\frac{-\gamma\sqrt{\eta}}{\bar{\gamma}\sqrt{\alpha_2}},
			\frac{-\gamma\delta_1}{\bar{\gamma}}, \frac{-\gamma\delta_2}{\bar{\gamma}},\frac{-\gamma c_1}{\bar{\gamma}}, \frac{-\gamma c_2}{\bar{\gamma}}\right),
\end{align}

\begin{align}\label{eqCDF}
		F_{\gamma}(\gamma) = \frac{\alpha_2^{m-\mu/2} \gamma^{\mu}}{\bar{\gamma}^{\mu}\Gamma(\mu+1)\alpha_1^m}
		\Phi_2^{(6)} \left(\frac{\mu}{2}, \frac{\mu}{2},-m, -m, m, m; \mu+1;\right. \left.\frac{-\gamma}{\bar{\gamma}\sqrt{\eta\alpha_2}}, 
			\frac{-\gamma\sqrt{\eta}}{\bar{\gamma}\sqrt{\alpha_2}},
			\frac{-\gamma\delta_1}{\bar{\gamma}}, \frac{-\gamma\delta_2}{\bar{\gamma}},\frac{-\gamma c_1}{\bar{\gamma}}, \frac{-\gamma c_2}{\bar{\gamma}}\right).
\end{align}

\hrulefill
\hrulefill
\vspace*{4pt}
\end{figure*}

\setcounter{equation}{1}

We will now provide a first-order characterization of the FB distribution in terms of its chief probability functions; {as we will later see, tractable analytical expressions are attainable for its MGF, PDF and CDF}. Hereinafter, we will consider the random variable $\gamma\overset{\Delta}{=}\nolinebreak\bar{\gamma}W/\bar{W}$, representing the instantaneous SNR at the receiver side. 

\begin{definition}{Let $\gamma$ be a random variable characterizing the instantaneous SNR for the physical model in (\ref{Model}). Then, $\gamma$ is said to follow a Fluctuating Beckmann (FB) distribution with mean ${\bar{\gamma}=\mathbb{E}[\gamma]}$ and non-negative real shape parameters $\kappa$, $\mu$, $m$, $\eta$ and $\varrho$, i.e. ${\gamma\sim\mathcal{FB}(\bar{\gamma}; \kappa, \mu, m,\eta,\varrho)}$, with}
\begin{align}\label{eqDef}
\kappa &= \frac{p^2+q^2}{\mu\left(\sigma_x^2+\sigma_y^2\right)}, & \varrho^2 &= \frac{p^2}{q^2}, &  \eta &= \frac{\sigma_x^2}{\sigma_y^2},
\end{align}
$\mu$ represents the number of clusters and $m$ accounts for the fluctuation of the {{LoS}} component.
\end{definition}
\vspace{3mm}
With the above definition, we now calculate the MGF of $\gamma$ in the following lemma.

\begin{lemma}
\label{lemma1}
Let ${\gamma\sim\mathcal{FB}(\bar{\gamma}; \kappa, \mu, m,\eta,\varrho)}$. Then, the MGF of $\gamma$ is given at the top of next page in \eqref{MGF_fin}.
\end{lemma}

\begin{proof}
See Appendix \ref{app1}.
\end{proof}

Lemma \ref{lemma1} provides a simple closed-form expression for the MGF of the FB fading distribution. From (\ref{MGF_fin}), we will now show that the PDF
and CDF of the FB fading distribution { have a similar functional form as the $\kappa$-$\mu$ shadowed fading distribution \cite{Paris2014}}.

\begin{lemma}
\label{lemma2}
Let ${\gamma\sim\mathcal{FB}(\bar{\gamma}; \kappa, \mu, m,\eta,\varrho)}$. Then, the PDF of $\gamma$ is given by \eqref{eqPDF} at the top of next page, 
where {$\delta_{x}$, $c_{x}$ and $\alpha_{x}$, with $x=\{1,2\}$} depend on the parameters of the FB distribution as described in the sequel, and $\Phi_2^{(n)}$ is the confluent form of the generalized Lauricella series defined in \cite[eq. 7.2, pp. 446]{erdelyi1937}.
\end{lemma}

\begin{proof}
Manipulating \eqref{MGF_fin} it is possible to write the MGF expression as follows
\begin{align}
	&M_{\gamma}(s) =
	\left( 1- \tfrac{\tfrac{\mu(1+\eta)(1+\kappa)}{2\eta\bar{\gamma}}}{s} \right)^{-\tfrac{\mu}{2}}\left( 1- \tfrac{\tfrac{\mu(1+\eta)(1+\kappa)}{2\bar{\gamma}}}{s} \right)^{-{\tfrac{\mu}{2}}}\times\notag \\
	& \tfrac{-1^{\mu}}{s^{\mu}}\tfrac{\alpha_2^{m-\mu/2}}{\bar{\gamma}^{\mu}\alpha_1^m}
	\left(1 - \tfrac{\delta_1}{\bar{\gamma}s}  \right)^m\left(1 - \tfrac{\delta_2}{\bar{\gamma}s}\right)^m 
	\left(1 - \tfrac{c_1}{\bar{\gamma}s}\right)^{-m}\left(1 - \tfrac{c_2}{\bar{\gamma}s}\right)^{-m}, \label{MGF_fac}
\end{align}
where $\delta_{1,2}$ are the roots of $\alpha_1s^2+\beta_1s + 1$ and $c_{1,2}$ the roots of $\alpha_2s^2+\beta_2s + 1$ with
\begin{align}
	\alpha_1 &= \frac{4\eta}{\mu^2(1+\eta)^2(1+\kappa)^2} + 
		\frac{2\kappa(\varrho^2+\eta)}{m (1+\varrho^2)\mu(1+\eta)(1+\kappa)^2}, \\
	\beta_1 &= \frac{-1}{1+\kappa}\left[ \frac{2}{\mu} + \frac{\kappa}{m} \right],  \\
	\alpha_2 &= \frac{4\eta}{\mu^2(1+\eta)^2(1+\kappa)^2}, \\
	\beta_2 &= \frac{-2}{\mu(1+\kappa)}.
\end{align}

The expression for the PDF can be derived from \eqref{MGF_fac} as $f_{\gamma}(\gamma) = \nolinebreak \mathcal{L}^{-1} \{ M_{\gamma}(-s) \}$ using \cite[eq. 9.55]{Srivastava1985}, yielding (\ref{eqPDF}).%
\end{proof}
\vspace{3mm}

\begin{lemma}
\label{lemma3}
Let ${\gamma\sim\mathcal{FB}(\bar{\gamma}; \kappa, \mu, m,\eta,\varrho)}$. Then, the CDF of $\gamma$ is given by \eqref{eqCDF} at the top of this page, 
\end{lemma}

\begin{proof}
Following the same steps as in the previous proof,  the CDF expression is given by $F_{\gamma}(\gamma) = \nolinebreak \mathcal{L}^{-1} \{ \frac{M_{\gamma}(-s)}{s} \}$, yielding (\ref{eqPDF}) directly from \cite[eq. 9.55]{Srivastava1985}.
\end{proof}
\setcounter{equation}{10}

The PDF and CDF of the FB distribution are given in terms of the multivariate $\Phi_2$ function, which also {appears} in other fading distributions in the literature \cite{Morales2010,Paris2014,Romero2016}. Apparently, and because it is defined as an $n$-fold infinite summation, its numerical evaluation may pose some challenges from a computational point of view. However, the Laplace transform of the $\Phi_2$ function has a comparatively simpler form in terms of a finite {product} of elementary functions, which becomes evident by inspecting the expression of the MGF in (\ref{MGF_fin}). Therefore, the $\Phi_2$ function can be evaluated by means of an inverse Laplace transform \cite{Martos2016,Abate1995}.

Note that the CDF and PDF of the received signal envelope can be directly derived from (\ref{eqPDF}) and (\ref{eqCDF}) straightforwardly through a change of variables. Thus, we get $f_R(R)=2R f_{\gamma}(R^2)$ and $F_R(R)=F_{\gamma}(R^2)$, with $\bar\gamma$ being replaced by $\Omega=E\{R^2\}$. 

As argued before, the FB distribution provides the unification of a large number of important fading distributions. These connections are summarized in Table \ref{Table1}, on which the parameters corresponding to the FB distribution are underlined in order to avoid confusion with the parameters of any of the distributions included as special cases. Notably, {the Beckmann distribution arises as a special case} of the more general FB distribution for $\mu=1$ and sufficiently large $m$. Thus, the additional degrees of freedom of the FB distribution also facilitates the anaytical characterization of the Beckmann distribution. Interestingly, when $\eta=1$ the effect of the parameter $\varrho$ vanishes; conversely, when setting $\varrho=1$ the effect of $\eta$ is still relevant. This is in coherence with the behavior of the Beckmann distribution as observed in \cite{Pena2017}.

\begin{table}[t]
\renewcommand{\arraystretch}{1.7}
\caption{Connections between the Fluctuating Beckmann fading model and other models in the literature. Note that setting $\kappa$=0 implies that $m$ and $\varrho$ vanish.}
\label{Table1}
\centering
\begin{tabular}
{c|c}
\hline
\hline
Channels  & Fluctuating Beckmann Fading Parameters\\
\hline
\hline
\multirow{1}{*}{One-sided Gaussian} &  $\underline{\kappa}=0$,\hspace{3mm}$\underline{\mu}=1$,\hspace{3mm}$\underline{\eta}=0$\hspace{3mm}\\
\hline
\multirow{1}{*}{Rayleigh} &  $\underline{\kappa}=0$,\hspace{3mm}$\underline{\mu}=1$,\hspace{3mm}$\underline{\eta}=1$\hspace{3mm}\\
\hline
\multirow{1}{*}{Nakagami-$m$} & $\underline{\kappa}=0$,\hspace{3mm}$\underline{\mu}=m$,\hspace{3mm}$\underline{\eta}=1$\hspace{3mm} \\
\hline
\multirow{1}{*}{Hoyt}  & $\underline{\kappa}=0$,\hspace{3mm}$\underline{\mu}=1$,\hspace{3mm}$\underline{\eta}=q$\hspace{3mm} \\
\hline
\multirow{1}{*}{$\eta$-$\mu$}& $\underline{\kappa}=0$,\hspace{3mm}$\underline{\mu}=\mu$,\hspace{3mm}$\underline{\eta}=\eta$\hspace{3mm} \\
\hline
Rice &   $\underline{\kappa}=K$,\hspace{3mm}$\underline{\mu}=1$,\hspace{3mm}$\underline{m}\rightarrow\infty$,\hspace{3mm}$\underline{\eta}=1$\hspace{3mm}$,\forall{\varrho}$\hspace{3mm}  \\
\hline
Symmetrical $\eta$-$\kappa$ &  $\underline{\kappa}=\kappa$,\hspace{3mm}$\underline{\mu}=1$,\hspace{3mm}$\underline{m}\rightarrow\infty$,\hspace{3mm}$\underline{\eta}=\eta$,\hspace{3mm}$\underline{\varrho}=\eta$\hspace{3mm}  \\
\hline
Asymmetrical $\eta$-$\kappa$  &  $\underline{\kappa}=\kappa$,\hspace{3mm}$\underline{\mu}=1$,\hspace{3mm}$\underline{m}\rightarrow\infty$,\hspace{3mm}$\underline{\eta}=\eta$,\hspace{3mm}$\underline{\varrho}=0$\hspace{3mm}  \\
\hline
Beckmann &  $\underline{\kappa}=K$,\hspace{3mm}$\underline{\mu}=1$,\hspace{3mm}$\underline{m}\rightarrow\infty$,\hspace{3mm}$\underline{\eta}=q$,\hspace{3mm}$\underline{\varrho}=r$\hspace{3mm}  \\
\hline
$\kappa$-$\mu$ &  $\underline{\kappa}=\kappa$,\hspace{3mm}$\underline{\mu}=\mu$,\hspace{3mm}$\underline{m}\rightarrow\infty$,\hspace{3mm}$\underline{\eta}=1$,\hspace{3mm}$\forall{\varrho}$\hspace{3mm}  \\
\hline
Rician Shadowed &  $\underline{\kappa}=\kappa$,\hspace{3mm}$\underline{\mu}=1$,\hspace{3mm}$\underline{m}=m$,\hspace{3mm}$\underline{\eta}=1$,\hspace{3mm}$\forall{\varrho}$\hspace{3mm}  \\
\hline
$\kappa$-$\mu$ shadowed &  $\underline{\kappa}=\kappa$,\hspace{3mm}$\underline{\mu}=\mu$,\hspace{3mm}$\underline{m}=m$,\hspace{3mm}$\underline{\eta}=1$,\hspace{3mm}$\forall{\varrho}$\hspace{3mm}  \\
\hline
\hline
\end{tabular}
\end{table}
\section{Second Order Statistics}
\label{SecondOrder}

First-order statistics such as the PDF, CDF or MGF provide valuable information about the statistical behavior of the amplitude (or equivalently power) of the received signal affected by fading. However, they do not incorporate information related to the dynamic behavior of the fading process, which is of paramount relevance in the context of wireless communications because of the relative motion of transmitter, receivers and scatterers due to mobility. In the literature, two metrics are used to capture the dynamics of a general random process: the level crossing rate (LCR), which measures how often the amplitude of the received signal crosses a given threshold value, and the average fade duration (AFD), which measures how long the amplitude of the received signal remains below this threshold \cite{Rice1944}.

The LCR of the received signal amplitude $R$ can be computed using Rice's formula \cite{Rice1944} as

\begin{align}
\label{eqRice}
N_R(u)=\int_0^{\infty}\dot{r}f_{R,\dot{R}}\left(u,\dot{r}\right)d\dot{r},
\end{align}
where $\dot{R}$ denotes the time derivative of the signal envelope and $f_{R,\dot{R}}(r,\dot{r})$ is the joint PDF of the received signal amplitude and its time derivative. Thus, in order to characterize the LCR of $R$, we must calculate the joint distribution of $R$ and $\dot{R}$. In our derivations, we will assume that the fluctuations in the diffuse part (i.e., {{NLoS}}) occur at a smaller scale compared to those of the {{LoS}} component in the fluctuating Beckmann fading model. This is the case, for instance, on which such {{LoS}} fluctuation can be associated to shadowing.

Let us express the squared signal envelope as
\begin{equation} R^2=R_1^2+R_2^2.\end{equation}
where $R_1$ and $R_2$ are defined as 
\begin{equation}R_1^2=\sum_{k=1}^\mu(X_k+\xi p_k)^2,\;\; R_2^2=\sum_{k=1}^\mu (Y_k+\xi q_k)^2,\end{equation}
Note that both variables, when conditioned to $\xi$, are independent. After normalizing by \mbox{$\Omega = {E[R^2]}$}, we have that $R_1$ and $R_2$ are distributed as a $\kappa$-$\mu$ random variables, with probability density function given by
\begin{equation}f_{R_k}(r_k)=\tfrac{\Omega^{\mu/4+1/2}}{\sigma^2_k(\xi d_k)^{\mu/2-1}}r_k^{\mu/2}e^{-\tfrac{\Omega r_k^2}{2\sigma^2_k}-\tfrac{\xi^2d_k^2}{2\sigma^2_k}}\text{I}_{\mu/2-1}\left(\tfrac{\Omega^{1/2}r_k\xi d_k}{\sigma_k^2}\right) \label{PDFR1}, \end{equation}
where $d_1^2=p^2=\sum_{k=1}^{\mu}p_k^2$, $d_2^2=q^2=\sum_{k=1}^{\mu}q_k^2$,  
and $\xi$ is a Nakagami-m distributed random variable with PDF given by
\begin{equation}
f_{\Xi}(\xi)=\frac{2m^m}{\Gamma(m)}\xi^{2m-1}\exp(-m\xi^2)
\end{equation} 

The derivative of $R$ with respect to time, $\dot{R}$, can be expressed as
\begin{equation}\dot{R}=\frac{\dot{R_1}R_1+\dot{R_2}R_2}{R}.\end{equation}
Conditioned to $R_1$, $R_2$ and $R$, the derivative of $R$ is a zero-mean Gaussian variable with variance 
\begin{equation}
\label{eq0000}\sigma^2_{\dot{R}}=\frac{\sigma^2_{\dot{R_1}}R^2_1+\sigma^2_{\dot{R_2}}R^2_2}{R^2}=\frac{\sigma^2_{\dot{R_1}}R^2_1+\sigma^2_{\dot{R_2}}R^2_2}{R^2_1+R^2_2}.\end{equation}
Hence, the distribution of $\dot{R}$ conditioned to $R_1$ and $R_2$ is
\begin{equation}f_{\dot{R}|R_1,R_2}= \frac{1}{\sqrt{2\pi\left(\frac{\sigma^2_{\dot{r_1}}r^2_1+\sigma^2_{\dot{r_2}}r^2_2}{r^2_1+r^2_2}\right)}}\cdot e^{-\frac{(r_1^2+r_2^2)\dot{r}^2}{2(\sigma^2_{\dot{r_1}}r^2_1+\sigma^2_{\dot{r_2}}r^2_2)}}, \label{PDFGauss}\end{equation}

The LCR can be obtained as
\begin{equation}\begin{split}
N_R(u)=&\int_0^\infty \dot{r} f_{R,\dot{R}}(u,\dot{r})d\dot{r}\notag\\=&\int_0^\infty \dot{r} \left( \int_0^u f_{\dot{R}| R,R1}(\dot{r},u,r_1)f_{R,R1}(u,r_1)dr_1\right)d\dot{r}\\
=& \int_0^u  f_{R,R_1}(u,r_1) \left(\int_0^\infty \dot{r} f_{\dot{R}|R,R_1}(\dot{r},u,r_1)d\dot{r}\right)dr_1.
\end{split}\end{equation}

Using the PDF of $\dot{R}$ conditioned to $R$ and $R_1$ in equation (\ref{PDFGauss}),
\begin{equation}\int_0^\infty \dot{r} f_{\dot{R}|R,R_1}(\dot{r},u,r_1)d\dot{r}=\sqrt{\frac{\sigma_{\dot{R}}^2}{2\pi}}.\end{equation}
The joint distribution of $R_1$ and $R_2$ can be obtained as
\begin{equation}
\begin{split}
&f_{R_1,R_2}(r_1,r_2)=\int_0^\infty f_{R_1|\xi}(r_1,\xi)f_{R_2|\xi}(r_2,\xi)f_\Xi(\xi)d\xi\\
=& \frac{2m^m\Omega^{\mu/2+1}}{\Gamma(m)\sigma^2_1\sigma_2^2 (p q)^{\mu/2-1}}r_1^{\mu/2}r_2^{\mu/2}\exp\left(-r_1^2\frac{\Omega}{2\sigma_1^2} - r_2^2\frac{\Omega}{2\sigma_2^2}\right) \times \\
&\int_0^\infty \left\{\xi^{2m-\mu+1}\exp\left(-\xi^2\left(\frac{p^2}{2\sigma_1^2}+\frac{q^2}{2\sigma_2^2}+m\right)\right)\right.\times\\ &\left.\text{I}_{\mu/2-1}\left(\frac{\Omega^{1/2}r_1\xi p}{\sigma_1^2}\right)\text{I}_{\mu/2-1}\left(\frac{\Omega^{1/2}r_2\xi q}{\sigma_2^2}\right) d\xi\right\}
\end{split}
\end{equation}
In the general case, the last integral cannot be solved analytically in closed-form, {to the best of the authors’ knowledge}. However, considering the special case of $q_k=0$ (the case of $p_k=0$ can be solved similarly) in which $R_1$ and $R_2$ are independent, the distribution of $f_{R,R_1}(r,r_1)$ can be obtained as
\begin{equation}\begin{split}
f_{R,R_1}(u,r_1)= &|J_{r_1,r_2}(u,r_1)|f_{R_1,R_2}(r_1,\sqrt{u^2-r_1^2})\\=&|J_{r_1,r_2}(u,r_1)|f_{R_1}(r_1)f_{R_2}(\sqrt{u^2-r_1^2})\\
=&\frac{m^m\Omega^{\mu}u\cdot r^{\mu-1}_1\cdot (u^2-r^2_1)^{\mu/2-1}}{2^{\mu-2}\Gamma^2(\mu/2)\sigma_1^\mu\sigma_2^\mu(\frac{p^2}{2\sigma^2_1}+m)^m} \times \\
&e^{-\frac{\Omega(u^2-r^2_1)}{2\sigma^2_2}-\frac{\Omega r^2_1}{2\sigma^2_1}} \cdot {}_1F_1\left(m,\mu/2;\frac{\Omega\frac{p^2}{4\sigma^4_1}}{\frac{p^2}{2\sigma^2_1}+m}r_1^2\right), \\
\end{split} \end{equation}
for $0 \leq r_1 \leq r$, where $|J_{r_1,r_2}(\cdot,\cdot)|$ denotes the Jacobian of the transformation of random variables. In this particular case, the LCR of the normalized envelope $R$ can be expressed as
\begin{equation}
\begin{split}
\label{LCRvar}
&N_R(u)=  \frac{m^m\Omega^{\mu}u}{2^{\mu-2}\Gamma^2(\mu/2)\sigma_1^\mu\sigma_2^\mu(\frac{p^2}{2\sigma^2_1}+m)^m \sqrt{2\pi}} \times \\
&\int_0^u \left(\sigma^2_{\dot{R_1}}r^2_1/u^2+\sigma^2_{\dot{R_2}}(1-r^2_1/u^2)\right)^{\frac{1}{2}}  (u^2-r^2_1)^{\mu/2-1}  \times\\&r^{\mu-1}_1 e^{-\frac{\Omega(u^2-r^2_1)}{2\sigma^2_2}-\frac{\Omega r^2_1}{2\sigma^2_1}} {}_1F_1\left(m,\tfrac{\mu}{2};\frac{\Omega\frac{p^2}{4\sigma^4_1}}{\tfrac{p^2}{2\sigma^2_1}+m}r_1^2\right)   dr_1 
\end{split}
\end{equation}
and after a change of variables we have
\begin{equation}
\begin{split}
&N_R(u)=\frac{m^m\Omega^{\mu}u^{(2\mu-1)}}{2^{\mu-1}\Gamma^2(\mu/2)\sigma_1^\mu\sigma_2^\mu(\frac{p^2}{2\sigma^2_1}+m)^m \sqrt{2\pi}} \cdot e^{-\frac{\Omega u^2}{2\sigma^2_2}}  \times \\
&\int_0^1 [\sigma^2_{\dot{R_2}}+(\sigma^2_{\dot{R_1}}-\sigma^2_{\dot{R_2}})x]^{\frac{1}{2}} (1-x)^{(\mu/2-1)} x^{(\mu/2-1)}\times\\&e^{-\frac{\Omega(\sigma^2_2-\sigma^2_1)u^2 x}{2\sigma^2_1\sigma^2_2}}{}_1F_1\left(m,\tfrac{\mu}{2};\tfrac{\Omega\frac{p^2}{4\sigma^4_1}}{\frac{p^2}{2\sigma^2_1}+m}u^2 x\right)dx
\end{split}
\end{equation}
where $\sigma^2_{\dot{R_1}}=\frac{-\ddot{\rho}(0)\sigma_1^2}{\Omega}$ and $\sigma^2_{\dot{R_2}}=\frac{-\ddot{\rho}(0)\sigma_2^2}{\Omega}$ from \eqref{eq0000}, and $\ddot{\rho}(0)$ is the second derivative of the autocorrelation function evaluated at $0$. 

Finally, \eqref{LCRvar} can be expressed in terms of the parameters of the FB distribution\footnote{Note that, because of the assumption of $q_k=0$, this implies that the parameter $\varrho\rightarrow\infty$. }, yielding
\begin{equation}
\label{eqLCRfinal}
\begin{split} & N_R(u)= \tfrac{m^m[\mu(1+\eta)(1+\kappa)]^{\mu-1/2}\sqrt{-\ddot{\rho}(0)}}{2^{\mu-1}\Gamma^2(\mu/2)\eta^{\mu/2}(\frac{\mu \kappa (1+\eta)}{2\eta}+m)^m \sqrt{2\pi}} \cdot u^{(2\mu-1)}  \times \\
&e^{-\mu/2 (1+\eta)(1+\kappa)u^2}\int_0^1 [1+(\eta-1)x]^{\frac{1}{2}} (1-x)^{(\mu/2-1)} x^{(\mu/2-1)}\times\\&e^{-\frac{\mu(1-\eta^2)(1+\kappa)}{2\eta}u^2 x}{}_1F_1\left(m,\mu/2;\frac{\frac{\kappa\mu^2(1+\eta)^2(1+\kappa)}{4\eta^2}}{\frac{\mu\kappa(1+\eta)}{2\eta}+m}u^2 x\right)dx
\end{split}\end{equation} 

{With the knowledge of the LCR and the CDF of the FB distribution}, the AFD can be directly obtained as

\begin{equation}
\label{eqAFDfinal}
T_R(u)=\frac{F_R(u)}{N_R(u)},
\end{equation}
where $F_R(u)$ is the CDF of the fading amplitude envelope derived in (\ref{eqCDF}), after a proper change of variables.

\section{Numerical Results and Applications}
\label{ApplicationsSection}

\subsection{First Order Statistics}
After attaining a full statistical characterization of the newly proposed Fluctuating Beckmann fading distribution, we aim to exemplify the influence of the parameters of this fading model over the distribution of the received amplitude. We will first focus on understanding the effect of the power imbalance in the {{LoS}} and {{NLoS}} components (i.e. the effect of $\varrho$ and $\eta$), since these are the two parameters that effectively extend the original $\kappa$-$\mu$ shadowed fading model to a more general case. Monte Carlo simulations are provided in order to double-check the validity of the derived expressions.

In Figs. \ref{PDF1} and \ref{PDF2}, the PDF of the received signal amplitude is represented for different values of {{NLoS}} power imbalance $\eta$ and {{LoS}} fluctuation severity $m$. The values $m=1$ and $m=10$ correspond to the cases of heavy and mild fluctuation of the {{LoS}} component, respectively. The parameter $\varrho$ is set to $\varrho^2=0.1$, indicating a moderately large {{LoS}} power imbalance, and $\mu=1$. Let us first focus on Fig. \ref{PDF1}, where
we set $\kappa=1$ to indicate a weak {{LoS}} scenario on which the {{LoS}} and {{NLoS}} power is the same. We observe that the effect of increasing $\eta$ causes the amplitude values {to be} more concentrated around its mean value. Besides, compared to the case of $\eta=1$ (i.e. the $\kappa$-$\mu$ shadowed fading distribution), the effect of having a power imbalance in the {{NLoS}} component clearly has an impact on the distribution of the signal envelope. Differently from the $\eta$-$\mu$ fading model, the behavior of the distribution with respect to $\eta$ is no longer symmetrical between $\eta\in[0,1]$ and $\eta\in[1,\infty)$ for a fixed $\varrho\neq1$. One interesting effect comes from the observation of the effect of increasing $\eta$: both setting $\eta=0.1$ or $\eta=10$ implies that the {{NLoS}} power is imbalanced by a factor of 10. However, it is evident that if this {{NLoS}} imbalance goes to the component associated with a larger {{LoS}} imbalance ($\eta=0.1$ since we have $\varrho^2=0.1$), this is way more detrimental for the received signal envelope than having the {{NLoS}} imbalance in the other component.

Fig. \ref{PDF2} now considers a strong {{LoS}} scenario on which $\kappa=10$. The rest of the parameters are the same ones as in the previous figure. Because the {{LoS}} component is now much more relevant, the effect of changing $m$ is more noticeable. We observe that for $m=10$, which corresponds to a mild fluctuation on the {{LoS}} component, the shape of the PDF is only slightly altered when changing $\eta$. Conversely, the shape of the PDF is more affected by $\eta$ for low values of amplitude when $m=1$. This is further exemplified in Fig. \ref{PDFbimodal}, on which a bimodal behavior is observed as the imbalance is reduced through $\varrho$ or $\eta$. When both $\{\varrho,\eta\}$ { decrease}, the in-phase components have considerably {less} power than the quadrature components. Because $\kappa$ is sufficiently large, the distribution will mostly fluctuate close to the {{LoS}} part of the quadrature component due to $m$, and the first maximum on the PDF in the low-amplitude region appears as the highly imbalanced in-phase component only is able to contribute in this region. We must note that this bimodal behavior does not appear in the original $\kappa$-$\mu$ shadowed or Beckmann distributions from which the FB distribution originates. Nevertheless, such bimodality indeed {shows} in other fading models such as the $\alpha$-$\eta$-$\kappa$-$\mu$ \cite{Yacoub2016}, the two-wave with diffuse power \cite{Durgin2002}, the fluctuating two-ray \cite{Romero2016} and some others \cite{Laureano2017,Beaulieu2014}.

\begin{figure}\centering
	\includegraphics[width=.99\columnwidth]{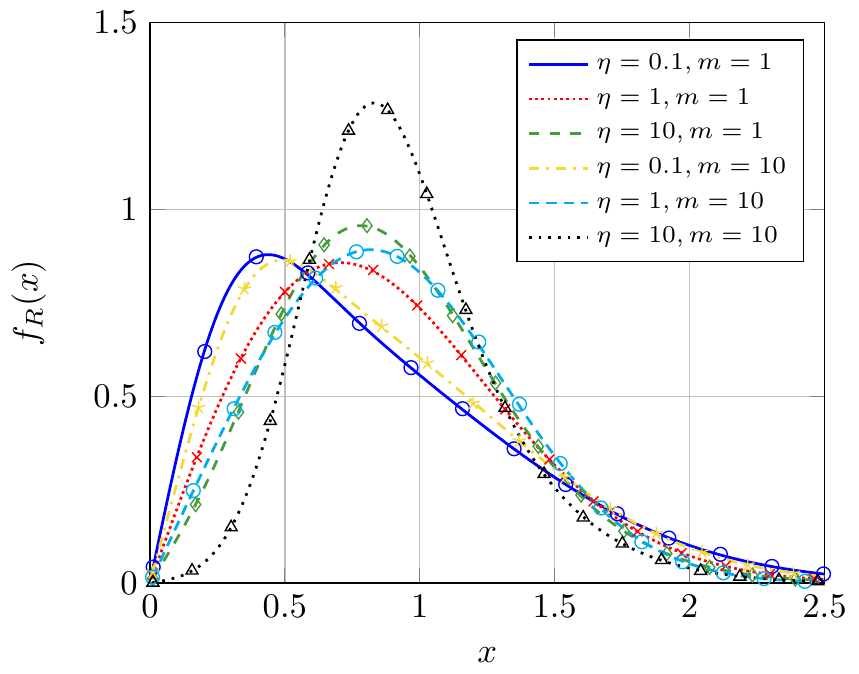}
	\caption{FB signal envelope distribution for different values of $\eta$ and $m$ in weak {{LoS}} scenario ($\kappa=1$) with $\varrho^2=0.1$, $\mu=1$ and
$\Omega=E\{R^2\}=1$. Solid lines correspond to the exact PDF, markers correspond Monte Carlo simulations.}
	\label{PDF1}
\end{figure}

\begin{figure}\centering
	\includegraphics[width=.99\columnwidth]{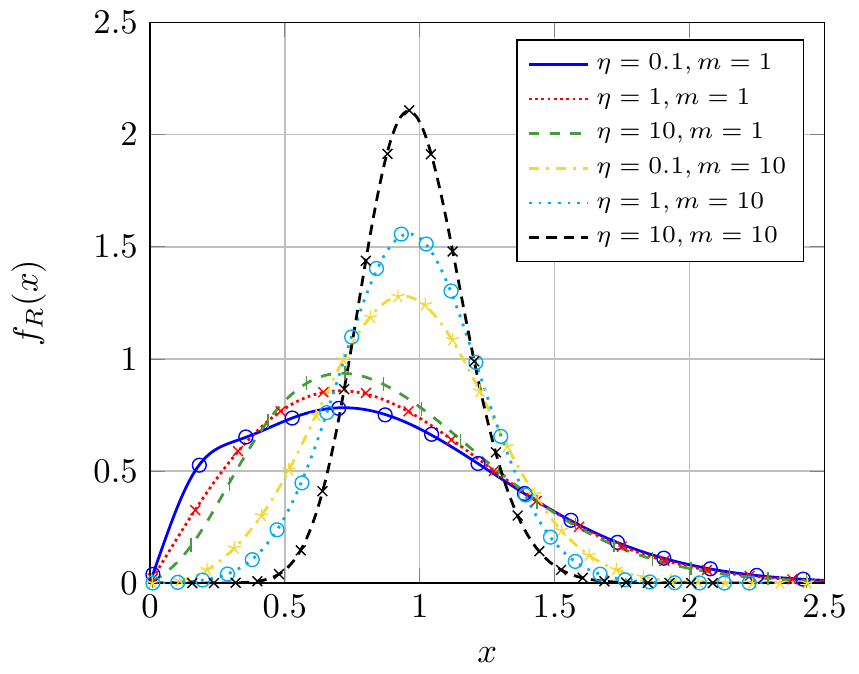}
	\caption{FB signal envelope distribution for different values of $\eta$ and $m$ in strong {{LoS}} scenario ($\kappa=10$) with $\varrho^2=0.1$, $\mu=1$ and
$\Omega=E\{R^2\}=1$. Solid lines correspond to the exact PDF, markers correspond Monte Carlo simulations.}
	\label{PDF2}
\end{figure}

\begin{figure}\centering
	\includegraphics[width=.99\columnwidth]{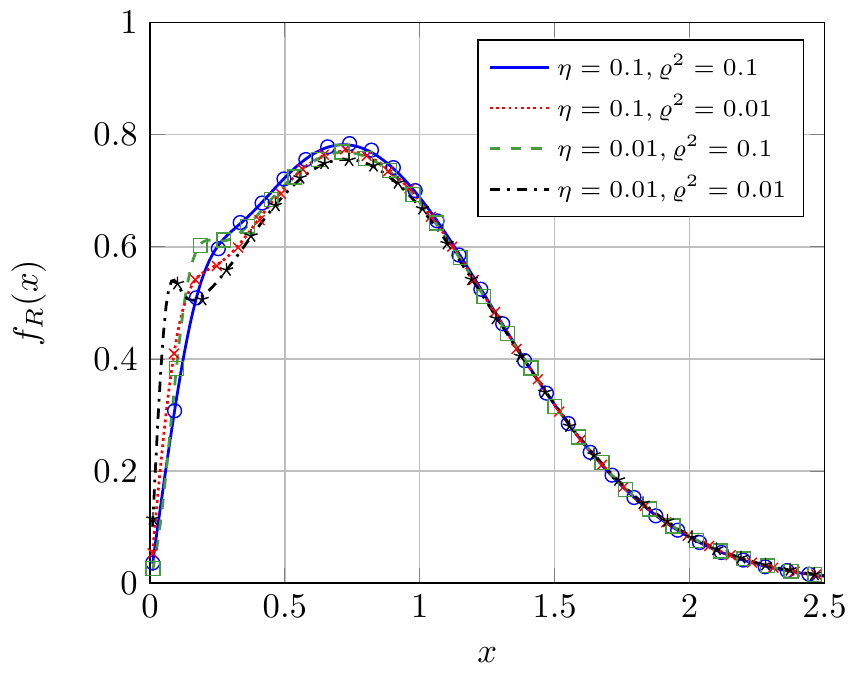}
	\caption{FB signal envelope distribution for different values of $\eta$ and $\varrho$ in strong {{LoS}} scenario ($\kappa=10$) with $m=1$, $\mu=1$ and
$\Omega=E\{R^2\}=1$. Solid lines correspond to the exact PDF, markers correspond Monte Carlo simulations.}
	\label{PDFbimodal}
\end{figure}

We represent in Figs. \ref{PDF3} and \ref{PDF4} the PDF of the received signal amplitude for different values of {{LoS}} power imbalance $\varrho$ and {{LoS}} fluctuation severity $m$. We first consider the weak {{LoS}} scenario with $\kappa=1$, and setting $\mu=2$ and $\eta=0.1$. We observe that low values of $\varrho$ and $m$ cause the amplitude values being more sparse. When the {{LoS}} component is stronger, i.e. $\kappa=10$ in Fig. \ref{PDF4} the effect of increasing $m$ (i.e. eliminating the {{LoS}} fluctuation) or $\varrho$ is more relevant.

\begin{figure}\centering
	\includegraphics[width=.99\columnwidth]{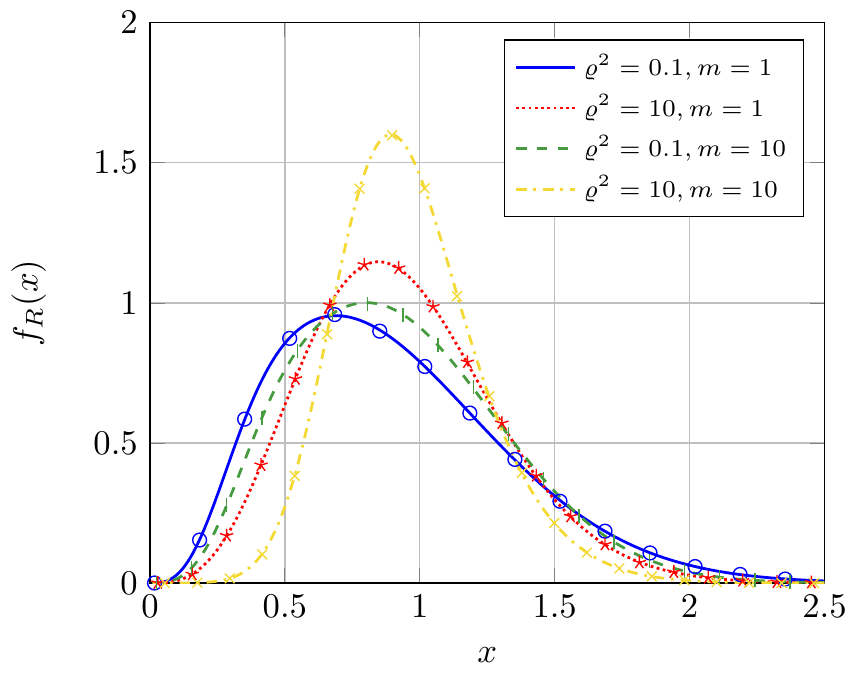}
	\caption{FB signal envelope distribution for different values of $\varrho$ and $m$ in weak {{LoS}} scenario ($\kappa=1$) with $\eta=0.1$, $\mu=2$ and
$\Omega=E\{R^2\}=1$. Solid lines correspond to the exact PDF derived, markers correspond Monte Carlo simulations.}
	\label{PDF3}
\end{figure}

\begin{figure}\centering
	\includegraphics[width=.99\columnwidth]{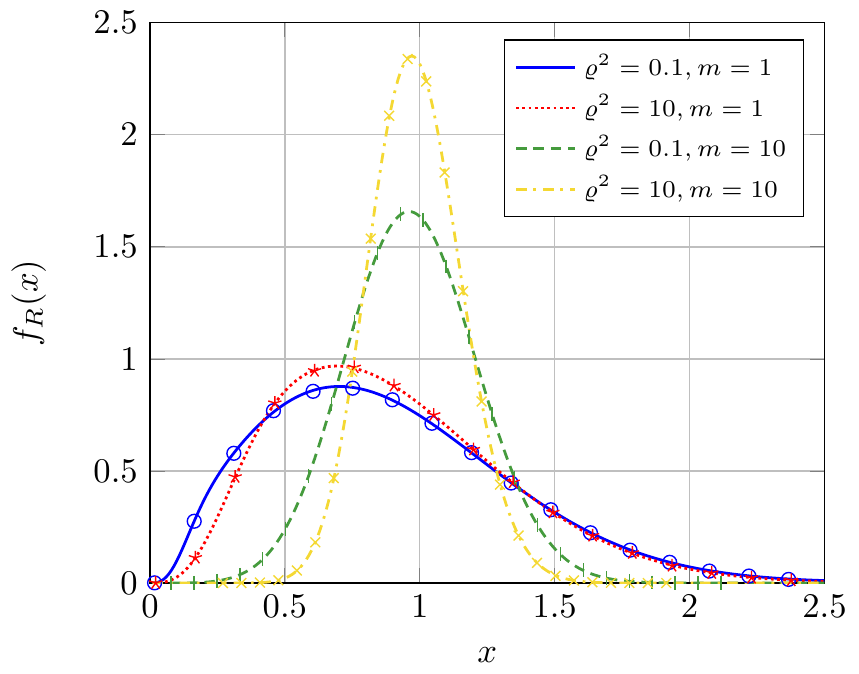}
	\caption{FB signal envelope distribution for different values of $\varrho$ and $m$ in strong {{LoS}} scenario ($\kappa=1$) with $\eta=0.1$, $\mu=2$ and
$\Omega=E\{R^2\}=1$. Solid lines correspond to the exact PDF, markers correspond Monte Carlo simulations.}
	\label{PDF4}
\end{figure}

Figs. \ref{CDF1} and \ref{CDF2} are useful to understanding the effect of the parameters $\varrho$ and $\eta$ over the CDF. Spefically, in Fig. \ref{CDF1} we compare the shape of the CDF in weak and strong {{LoS}} scenarios as $\eta$ varies. The {{LoS}} fluctuation parameter is set to $m=10$ in order to eliminate its influence, whereas $\varrho^2=0.1$ and $\mu=1$. We observe that increasing either $\eta$ or $\kappa$ makes the slope of the CDF rise close to $x=1$. A similar observation can be made when inspecting Fig. \ref{CDF2}. We see that having the {{LoS}} and {{NLoS}} imbalances in the same component ($\varrho^2=\eta=0.1$) is more detrimental for the signal envelope, and the probability of having very low values of signal level is higher.

\begin{figure}\centering
	\includegraphics[width=.99\columnwidth]{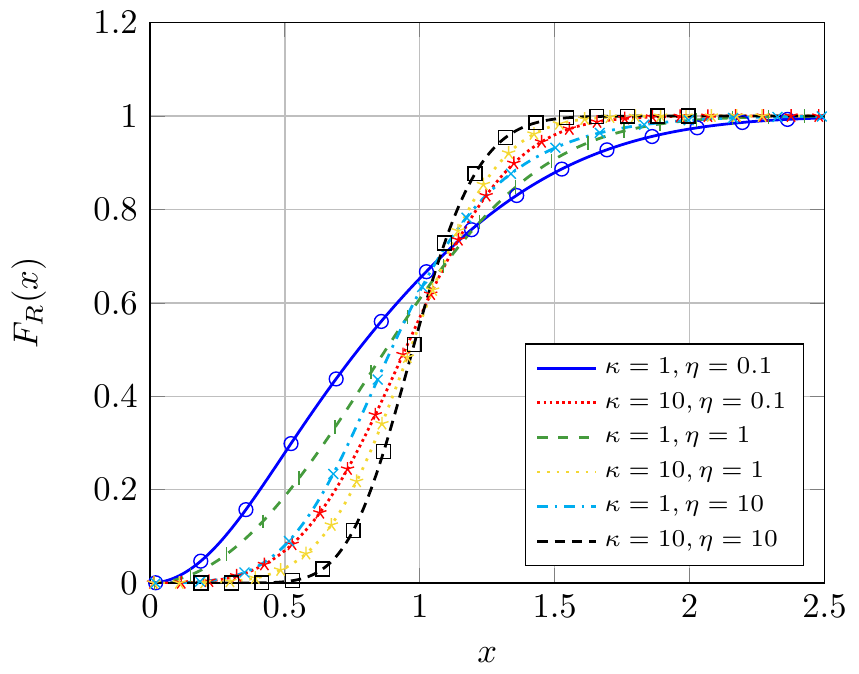}
	\caption{FB signal envelope CDF for different values of $\kappa$ and $\eta$ with $\varrho^2=0.1$, $\mu=1$, $m=10$ and
$\Omega=E\{R^2\}=1$. Solid lines correspond to the exact CDF, markers correspond Monte Carlo simulations.}
	\label{CDF1}
\end{figure}

\begin{figure}\centering
	\includegraphics[width=.99\columnwidth]{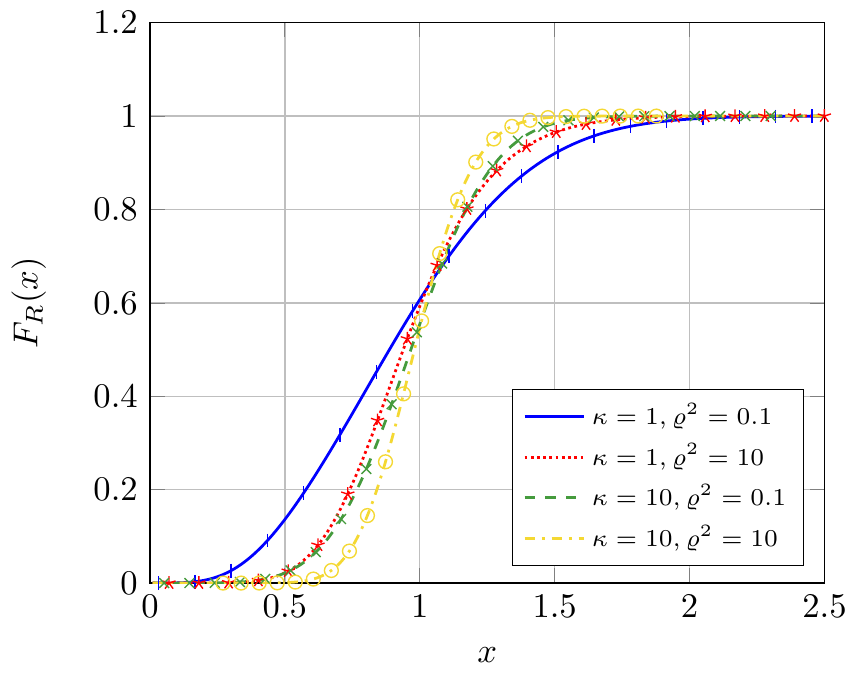}
	\caption{FB signal envelope CDF for different values of $\kappa$ and $\varrho$ with $\eta=0.1$, $\mu=2$, $m=10$ and
$\Omega=E\{R^2\}=1$. Solid lines correspond to the exact CDF, markers correspond Monte Carlo simulations.}
	\label{CDF2}
\end{figure}

\subsection{Second Order Statistics}

We will now investigate the effect of the FB fading parameters on the second-order statistics of the distribution. We assume that a time variation of the diffuse component according to Clarke's correlation model \cite{Stuber2011} with maximum Doppler shift $f_d$; this implies that $\sqrt{-\ddot{\rho}}=\sqrt{2}f_d\pi$ \cite{Ramos2009}. As argued in Section \ref{SecondOrder}, we consider that $\varrho\rightarrow\infty$ and hence the LCR and AFD are given by (\ref{eqLCRfinal}) and (\ref{eqAFDfinal}). Monte Carlo simulations are also included, by generating a sampled fluctuating Beckmann random process with sampling period $T_s>>f_d$ in order to avoid missing level crossings at very low threshold values \cite{Lopez2012}.

Fig. \ref{LCRfig} represents the LCR vs the normalized threshold for different sets of fading parameter values. When increasing $\mu$, i.e. the number of multipath clusters, the number of crossings at very low threshold values is drastically reduced. Similarly, the number of crossings in this region grows when reducing $\kappa$ or increasing $\eta$. This latter effect is coherent with the fact that $\varrho\rightarrow\infty$ in this case, so that having a value of $\eta<1$ is beneficial in terms of fading severity. Thus, the maximum number of crossings for low threshold values in the investigated scenarios is attained for low $\mu$ and $\kappa$, and large $\eta$.

\begin{figure}\centering
	\includegraphics[width=.99\columnwidth]{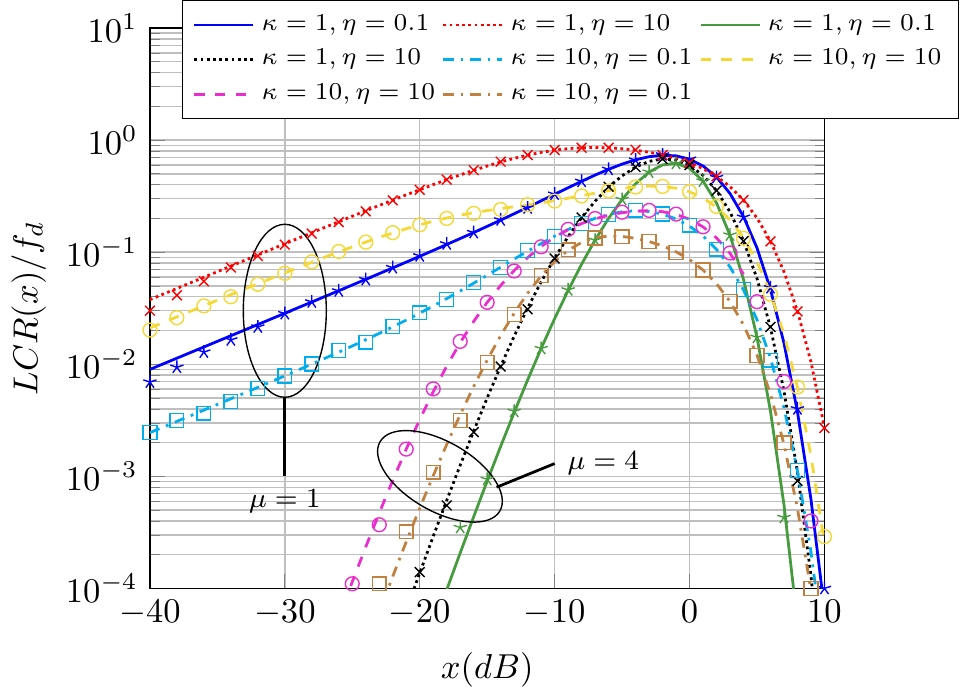}
	\caption{Normalized LCR vs threshold value $x$ (dB) normalized to $\Omega$ for different values of $\kappa$, $\eta$ and $\mu$, with $m=1$ and $\varrho\rightarrow\infty$. Solid lines correspond to the exact LCR, markers correspond Monte Carlo simulations.}
	\label{LCRfig}
\end{figure}

Fig. \ref{AFDfig} represents the AFD vs the normalized threshold for the same set of fading parameter values as in Fig. \ref{LCRfig}. Interestingly, we see that the duration of deep fades is not affected by $\eta$. We also observe that a larger AFD is associated with a lower value of $\mu$ and a larger value of $\kappa$; this is in coherence with the observations in \cite{Cotton2007} for the particular case of the $\kappa$-$\mu$ fading model.

\begin{figure}\centering
	\includegraphics[width=.99\columnwidth]{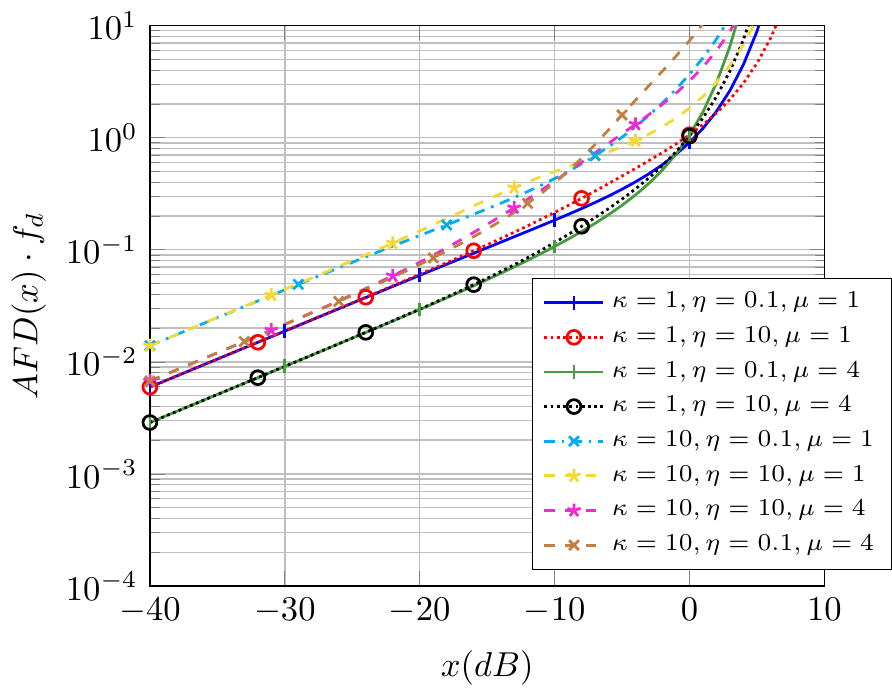}
	\caption{Normalized AFD vs threshold value $x$ (dB) normalized to $\Omega$ for different values of $\kappa$, $\eta$ and $\mu$, with $m=1$ and $\varrho\rightarrow\infty$. Solid lines correspond to the exact AFD, markers correspond Monte Carlo simulations.}
	\label{AFDfig}
\end{figure}

\begin{figure*}[!t]
\setcounter{equation}{27}

\normalsize
\begin{align}
\label{BER1}
	P_s(\bar{\gamma})=&\frac{1}{2\left(1+\frac{2\eta}{\mu(1+\eta)(1+\kappa)}\bar{\gamma} \right)^{\mu/2} \left(1+\frac{2}{\mu(1+\eta)(1+\kappa)}\bar{\gamma} \right)^{\mu/2}}\times\nonumber\\&\left[1+\frac{1}{m}\left(\frac{\mu\kappa\left( \frac{\varrho^2}{1+\varrho^2}\right)(1+\eta)\bar{\gamma}}{(1+\eta)(1+\kappa)\mu + 2\eta\bar{\gamma}}\right.\right.\left.\left.+  \frac{\mu\kappa\left( \frac{1}{1+\varrho^2}\right)(1+\eta)\bar{\gamma}}{(1+\eta)(1+\kappa)\mu + 2\bar{\gamma}}\right) \right]^{-m}
\end{align}

\setcounter{equation}{29}
\begin{align}
\label{BER2}
	P_s(\bar{\gamma})=&\sum_{n=1}^{N-1}(-1)^{n+1}\binom{N-1}{n}\frac{1}{n+1} \times\frac{1}{\left(1+\frac{2n\eta}{(n+1)\mu(1+\eta)(1+\kappa)}\bar{\gamma} \right)^{\mu/2}}\times
	 \frac{1}{\left(1+\frac{2n}{(n+1)\mu(1+\eta)(1+\kappa)}\bar{\gamma} \right)^{\mu/2}}  \notag \\
	&\times\left[1+\frac{1}{m}\left(\frac{\mu\kappa\left( \frac{\varrho^2}{1+\varrho^2}\right)(1+\eta)\bar{\gamma}\left(\frac{n}{n+1}\right)}{(1+\eta)(1+\kappa)\mu + 2\eta\bar{\gamma}\left(\frac{n}{n+1}\right)}\right.\right. \left.\left.+  \frac{\mu\kappa\left( \frac{1}{1+\varrho^2}\right)(1+\eta)\bar{\gamma}\left(\frac{n}{n+1}\right)}{(1+\eta)(1+\kappa)\mu + 2\bar{\gamma}\left(\frac{n}{n+1}\right)}\right) \right]^{-m}
\end{align}

\hrulefill
\hrulefill
\vspace*{4pt}
\end{figure*}

\subsection{Error Probability Analysis}
\setcounter{equation}{25}

We now exemplify how the performance analysis of wireless communication systems operating under FB fading can be carried out. For the sake of simplicity, we here focus on the symbol error probability (SEP) analysis for a number of well-known modulation schemes.

The SEP in the presence of fading is known to be given by
\begin{equation}
	P_s(\bar{\gamma}) = \int_0^{\infty} P_{AWGN}(\gamma)f_{\gamma}(\gamma)d\gamma
\end{equation}
where $P_{AWGN}(\gamma)$ is the symbol error probability in the AWGN case \cite[eq. 8.85]{AlouiniBook}. When using coherent DBPSK (Differential Binary Phase-.Shift Keying) modulation, the SEP can be expressed in terms of the MGF as 
\begin{equation}
P_s(\bar{\gamma}) =\nolinebreak \frac{1}{2}M_{\gamma}(s)\rvert_{s=-1}.
\end{equation}
Thus, the SEP of DPBSK when assuming the FB fading model is given in \eqref{BER1} at the top of this page.
%

\setcounter{equation}{28}

In the case of assuming orthogonal  $M$-ary FSK (Frequency-Shift Keying) signals and non-coherent demodulation, the symbol error probability over AWGN channels is given in \cite[eq. 8.67]{AlouiniBook} as
\begin{equation}
	P_s(\bar{\gamma}) = \sum_{n=1}^{M-1}(-1)^{n+1}\binom{M-1}{n}\left.\frac{1}{n+1}M_{\gamma}(s)\right|_{s=\frac{-n}{n+1}},
\end{equation}
yielding the following expression when assuming the FB fading model in \eqref{BER2}.
%

The SEP is evaluated in Fig. \ref{SEP_fig}, assuming coherent DBPSK, and non-coherent $2$-FSK and $4$-FSK. We observe that the SEP performance of DBPSK is much better than the non-coherent schemes, especially when the fading severity is reduced (i.e. large $\kappa$ and $\eta$, for $\varrho<1$)

\begin{figure}\centering
	\includegraphics[width=.99\columnwidth]{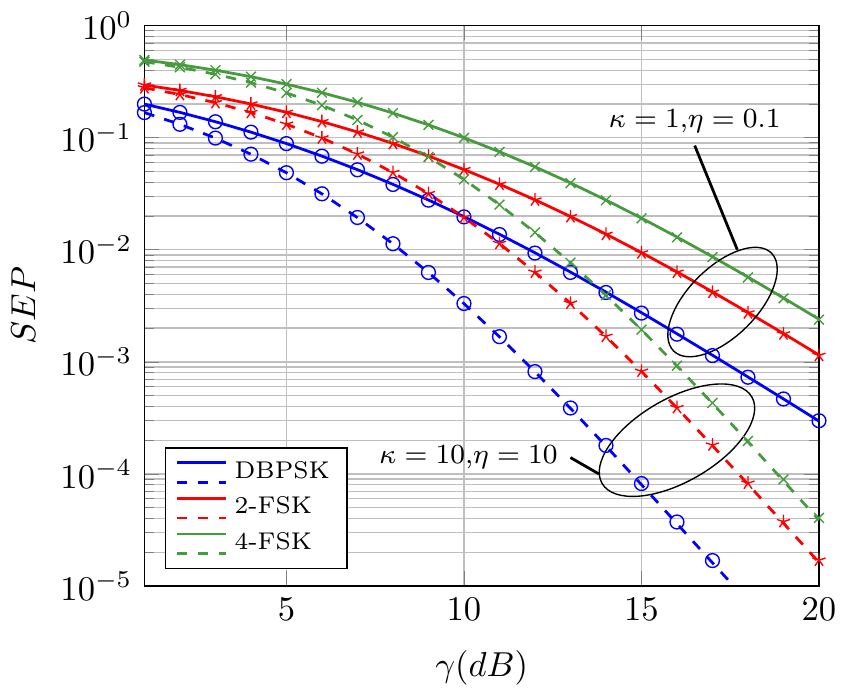}
	\caption{SEP vs. $\gamma$ for different values of $\kappa$ and $\eta$ and different modulation schemes. Parameter values are $m=4$, $\mu=2$ and $\varrho^2=0.2$. Solid lines correspond to the exact SEP, markers correspond Monte Carlo simulations.}.
	\label{SEP_fig}
\end{figure}
\setcounter{equation}{30}

\section{Conclusion}
\label{ConclusionSection}

We presented an extension of the $\kappa$-$\mu$ shadowed fading distribution, by including the effects of power imbalance between the {{LoS}} and {{NLoS}} components through two additional parameters, $\varrho$ and $\eta$, respectively. This generalization also includes the classical and notoriously unwieldy Beckmann fading distribution as special case, with the advantage of admitting a relatively simple analytical characterization when compared to state-of-the-art fading models. Thus, we are able to unify a wide set of fading models in the literature under the umbrella of a more general model, for which we suggest the name of Fluctuating Beckmann fading model.

We observed that when the {{LoS}} and {{NLoS}} imbalances are both large for the same component (i.e. $\varrho<1$ and $\eta<1$ for the in-phase component, or $\varrho>1$ and $\eta>1$ for the quadrature component), the fading severity is increased. Conversely, when the {{LoS}} imbalance is larger in one component (e.g. $\varrho<1$) it is beneficial that its {{NLoS}} part has {less} power (i.e. $\eta>1$ in this case) in order to reduce fading severity. Strikingly and somehow counterintuitively, the FB distribution exhibits a bimodal behavior in some specific scenarios, unlike the distributions from which it originates.

\appendices

\section{Proof of Lemma I}
\label{app1}
Let us consider the physical model in \eqref{Model}. Specializing for $\mu = \nolinebreak 1$, the conditional MGF of the signal power W given $\xi$ follows a Beckmann distribution with MGF given by \cite[eq. 2.38]{AlouiniBook}
\begin{align}
\label{MGF_Beckmann}
M_{W}(s | \xi)=&\frac{1}{(1-2\sigma^2_xs)^{1/2}(1-2\sigma^2_ys)^{1/2}} \times\notag\\ &\exp{\left( \frac{p_1^2\xi s}{1 - 2\sigma_x^2 s} + 
													\frac{q_1^2\xi s}{1 - 2\sigma_y^2 s}\right)}
\end{align}

Since the Gaussian processes within \eqref{Model} are mutually independent, then the conditional moment-generating function of the FB distribution can be obtained by multiplying the $\mu$ terms of the sum. Thus, the conditional MGF of the signal power W is given by
\begin{align}
\label{MGF_Ext_Cond}
M_{W}(s | \xi)=&\frac{1}{(1-2\sigma^2_xs)^{\mu/2}(1-2\sigma^2_ys)^{\mu/2}} \times\notag\\&\exp{\left( \frac{p^2\xi s}{1 - 2\sigma_x^2 s} + 
													\frac{q^2\xi s}{1 - 2\sigma_y^2 s}\right)}
\end{align}
where $p^2=\sum_{i=1}^{\mu}p_i^2$ and $q^2=\sum_{i=1}^{\mu}q_i^2$.

With the definitions in (\ref{eqDef}), the conditional MGF in \eqref{MGF_Ext_Cond} can be rewritten as:
\begin{align}
\label{MGF_new_param}
	&M_{\gamma}(s | \xi)=\frac{1}{\left(1-\frac{2\eta}{\mu(1+\eta)(1+\kappa)}\bar{\gamma s} \right)^{\mu/2} \left(1-\frac{2}{\mu(1+\eta)(1+\kappa)}\bar{\gamma s} \right)^{\mu/2}}  \notag \\
	&\times exp\left(\frac{\mu\kappa\left( \frac{\varrho^2}{1+\varrho^2}\right)(1+\eta)\xi\bar{\gamma}s}{(1+\eta)(1+\kappa)\mu - 2\eta\bar{\gamma} s} +  \frac{\mu\kappa\left( \frac{1}{1+\varrho^2}\right)(1+\eta)\xi\bar{\gamma}s}{(1+\eta)(1+\kappa)\mu - 2\bar{\gamma} s}\right)
\end{align}
Finally, the unconditional MGF for the FB fading model can be obtained by averaging \eqref{MGF_new_param} as
\begin{equation}
	M_{\gamma}(s)=\int_0^{\infty} M_{\gamma}(s | \xi) \, f_{\xi}(\xi) \, d\xi
\end{equation}
where $f_{\xi}(\xi)$ is the Nakagami-$m$ PDF, yielding (\ref{MGF_fin})
\vspace{10mm}
\bibliographystyle{IEEEtran}
\bibliography{bibjavi}

%
%
%
%
%
%


\end{document}